\documentclass[runningheads,a4paper]{llncs}
\usepackage[T1]{fontenc}
\usepackage{color}
\usepackage{hyperref}
\usepackage[dvipsnames]{xcolor}
\usepackage{graphicx}
\usepackage{amssymb,mathrsfs,setspace,amsmath}

\usepackage{algorithm,algpseudocode}
\usepackage{subcaption}
\usepackage{thmtools,thm-restate}
\usepackage{thmtools}

\usepackage{epsfig}
\input xy
\xyoption{all}
\usepackage{float}
\usepackage{pgfmath}
\usepackage{orcidlink}
\usepackage{tikz}
 \usepackage{subcaption}
\newcommand{\lsc}{\mathcal{LSC}}

\begin{document}
\title{Separation of Unconscious Robots with Obstructed Visibility
}
\titlerunning{Separation of Unconscious Robots with Obstructed Visibility}
\author{}
\institute{}
%
%
\author{Prajyot Pyati \and
Navjot Kaur \and
Saswata Jana \orcidlink{0000-0003-3238-8233} \thanks{Supported by  Prime Minister's Research Fellowship (PMRF) scheme of the Govt. of India (PMRF-ID: 1902165)} \and
Adri Bhattacharya \orcidlink{0000-0003-1517-8779} \thanks{Supported by CSIR, Govt. of India, Grant Number: 09/731(0178)/2020-EMR-I} \and
Partha Sarathi Mandal\orcidlink{0000-0002-8632-5767}}
\authorrunning{Pyati et al.}
\maketitle              
\begin{abstract}

We study a recently introduced \textit{unconscious} mobile robot model, where each robot is associated with a \textit{color}, which is visible to other robots but not to itself.
The robots are autonomous, anonymous, oblivious and silent, operating in the Euclidean plane under the conventional \textit{Look-Compute-Move} cycle.
A primary task in this model is the \textit{separation problem}, where unconscious robots sharing the same color must separate from others, forming recognizable geometric shapes such as circles, points, or lines.
All prior works model the robots as \textit{transparent}, enabling each to know the positions and colors of all other robots.
In contrast, we model the robots as \textit{opaque}, where a robot can obstruct the visibility of two other robots, if it lies on the line segment between them.
Under this obstructed visibility, we consider a variant of the separation problem in which robots, starting from any arbitrary initial configuration, are required to separate into concentric semicircles.
We present a collision-free algorithm that solves the separation problem under a semi-synchronous scheduler in $O(n)$ epochs, where $n$ is the number of robots.
The robots agree on one coordinate axis but have no knowledge of $n$.
\keywords{Opaque  \and Mobile robots \and  Separation  \and Unconscious robots}
\end{abstract}

\section{Introduction}

In the domain of distributed computing with a swarm of mobile robots, the robots are assumed to be \emph{oblivious} \cite{flocchini2012distributed}, and predominantly thought of as point objects \cite{cohen2008local}. These robots operate in a classical Look-Compute-Move (LCM) cycle.
Once activated, a robot performs a \emph{Look} by taking a snapshot of its surroundings. Based on this snapshot, it computes a destination during the \emph{Compute} phase, and finally, in the \emph{move} phase, it remains in place or moves to the computed destination.
The robots are \emph{autonomous} (no central control),  \emph{anonymous} (have no identifiers to distinguish among themselves), and \emph{oblivious} (do not remember any past actions or positions).
They work collaboratively by executing the same algorithm to achieve some global task, for example, pattern formation \cite{bose2021arbitrary} \cite{flocchini2008}, mutual visibility \cite{di2014mutual} \cite{pramanick2024mutual}, gathering \cite{agmon2006fault}, dispersion \cite{augustine2018dispersion}, etc.
In the conventional $\mathcal{OBLOT}$ model \cite{flocchini2008}, besides being oblivious, the robots are also silent, meaning they cannot communicate explicitly among themselves.

In general, robots communicate via externally visible lights (colors) \cite{das2016autonomous,flocchini2012distributed}, which they flash from a predefined color palette.
A robot displays its light to transmit a message and observes others’ lights to receive messages.
In this context, three major models have been studied.
Firstly, the $\mathcal{LUMI}$ model, where a robot can see the color of all the robots, including itself. Secondly, the $\mathcal{FSTA}$ model, where a robot can only see the color of its own, making it equivalent to  $\mathcal{OBLOT}$ model with some persistent memory (silent but not oblivious).
Thirdly, the recently introduced \textit{unconscious} colored robot model \cite{flocchini2025asynchronous}, \cite{seike2023}, where the robots are unaware of their own color, but can observe the colors of other robots.

One fundamental challenge in this unconscious colored robot model is the problem of separation, where the robots of same color must separate from others and form a group in terms of some recognizable geometric shapes such as points, lines or circles.
Seike and Yamauchi \cite{seike2023} were the first to study this problem.
They show an impossibility result: separation into points cannot be achieved due to the symmetry of the initial configuration. They also proposed an algorithm that solves the separation problem into circles using oblivious robots having common chirality (robots agree on the clockwise direction) under $\mathcal{SSYNC}$.
In $\mathcal{SSYNC}$, a non-empty set of robots gets activated simultaneously and performs their LCM cycle in sync.
Later on, Flocchini et al. \cite{flocchini2025asynchronous} solved the problem of separation into lines using the robots of similar properties as in \cite{seike2023} but under $\mathcal{ASYNC}$ scheduler and robots agree on one axis.
In $\mathcal{ASYNC}$, robots execute their LCM cycles at arbitrary times, each within a finite duration.
In both $\mathcal{SSYNC}$ and $\mathcal{ASYNC}$, the time is measured in terms of \textit{epochs}, defined as the smallest interval in which every robot completes at least one full LCM cycle.

Surprisingly, both of the previously mentioned works \cite{flocchini2025asynchronous}, \cite{seike2023} are considered the robot as \textit{transparent}, meaning that a robot never obstructs the visibility of other robots, and can always see the positions and colors of all other robots.
In this paper, we consider robots to be \textit{non-transparent} (or \textit{opaque}) \cite{bose2021arbitrary}, \cite{di2014mutual}, meaning that if a robot lies on the line segment between two others, it obstructs their mutual visibility.
Due to obstructed visibility, the robots are prone to collisions, which may affect their hardware. Moreover, they lack knowledge of the total number of robots, making signalling and coordination significantly more challenging in solving the separation problem. In this paper, we address these challenges and show that robots in $O(n)$ epochs, achieve a separated configuration without collision in finite time. To the best of our knowledge, this is the first work to introduce opaqueness in the separation problem for unconscious robots.
Our desired separated configuration is the concentric semicircles. Agreement on one axis allows us to fix the diameter line of these semicircles, which is perpendicular to the agreed axis. Hence, it is sufficient to assume that there are at least two robots of each color. The semicircular configuration naturally extends to a circular one if each color has at least three robots. The smaller lower bound on the number of robots thus motivates the choice of a semicircular formulation for the separated configuration. Moreover, the semicircular configuration breaks the perfect symmetry of the circular one while maintaining balance along the diameter. This is particularly important when robots must preserve an open space and move coherently toward that direction (for instance, exploring or guarding the open half-space opposite to the semicircular arc).

\noindent {\bf Our Contributions:} 
In this paper, we extend the study of unconscious robots by addressing the separation problem into concentric semicircles under $\mathcal{SSYNC}$ scheduler. We propose an algorithm, \textsc{Con-SemCirc-Separation}, that achieves this separation in $O(n)$ epochs without collision (Theorem \ref{thm:final}), where $n$ is the total number of robots. The robots are anonymous, silent, and oblivious, with no knowledge of $n$, but they agree on one coordinate axis.
This paper is the first to introduce opaqueness into the study of unconscious robots, adding significant complexity to the separation task, which our algorithm successfully resolves.
Table \ref{tab:comparioson} compares our work with most related works in the literature.

\begin{table}[H]
\centering
\begin{tabular}{|c|c|c|c|}
\hline
Papers & SSS 2023\cite{seike2023} & IJNC 2025\cite{flocchini2025asynchronous} & This Paper  \\ \hline
Opacity & $\times$ & $\times$ & $\checkmark$ \\ \hline
Separation Type & Concentric Circles & Lines & Concentric Semicircles \\ \hline
Scheduler & $\mathcal{SSYNC}$ & $\mathcal{ASYNC}$ & $\mathcal{SSYNC}$ \\ \hline
Chirality & $\checkmark$ & $\times$ & $\times$ \\ \hline
One-axis Agreement & $\times$ & $\checkmark$ & $\checkmark$\\ \hline
\end{tabular}
\caption{Comparing our model and result with the literature}
\label{tab:comparioson}
\vspace{-10mm}
\end{table}

\section{Model, Preliminaries and Problem Definition}

\textbf{Robots:} We consider a set of \( n \) \textit{anonymous} (no unique identifier), \textit{autonomous} (no external control), \textit{homogeneous} (run the same algorithm), and \textit{silent} (no explicit mode of communication) point robots \( R = \{r_1, r_2, \ldots, r_n\} \).
The robots operate in the Euclidean plane \( \mathbb{R}^2 \) and are initially located at distinct positions.
Let \( p_j(t) \) represent the position of the robot \( r_j \) at time \( t \).
Each robot \( r_j \) is associated with a value \( c_j \), called its \emph{color}, chosen from a totally ordered set \( H \), which is known to all the robots.
Let \( C \subseteq H \) be the subset of colors currently held by the robots, and let \( |C| = k \) denote the number of distinct colors in \( C \). Each color in \( C \) is assigned to at least two robots.
The configuration of the robots at time \( t \) is the set of tuples comprising the robot's positions and their colors which is denoted by the set:
$P(t) = \{(p_1(t), c_1), (p_2(t), c_2), \dots, (p_n(t), c_n)\}$.
When the time $t$ is clear from context, we abuse the notation $r_j$ to refer to the current position of robot $r_j$.
Besides the previous properties, the robots are \textit{opaque}, meaning that if a robot $r_j$ lies on the line segment connecting two other robots, it blocks the mutual view of the other two robots.
They are \textit{oblivious}, which means they have no memory to remember anything from the past executions of the algorithm.
They also have no knowledge about the total number of robots $n$ in the system.
They have the agreement on one-axis. Without loss of generality, we assume it is on the $y$-axis, that is, all robots agree on the positive direction of the $y$-axis but may disagree on the orientation of the $x$-axis.
Finally, the robots are \textit{unconscious} of their own color, they can see the colors of other robots but don't know which color they themselves hold.

\noindent \textbf{Activation Cycle:}
Each robot operates in a classical \emph{Look-Compute-Move} (LCM) cycle.
In the \emph{Look} phase, the robot observes its surroundings and takes a snapshot of the current configuration. 
We denote the snapshot observed by robot \( r_j \) at time \( t \) as \( Z_j(t) \), which is the set of tuples consisting of all visible robot positions and their colors.
In the \emph{Compute} phase, the robot runs the same deterministic algorithm \( \psi \), which takes the snapshot \( Z_j(t) \) as input and decides its destination or remain stationary.
Finally, in the \emph{Move} phase, the robot moves to the computed destination. The movement is \textit{rigid}, where a robot never stops before it reaches its destination.

\noindent \textbf{Activation Scheduler and Run Time:}
The robots are activated under a \textit{semi-synchronous} (in short $\mathcal{SSYNC}$) scheduler, where a non-empty set of robots is activated simultaneously at any given time step. 
The set of robots, that are activated together at a given time, execute their LCM cycle in sync, while the others remain inactive until the completion of the cycle.
Time is measured in terms of \textit{epochs}, where an epoch is the smallest time interval in which every robot is activated and
completes one full LCM cycle at least once.

\begin{definition}({\texttt{Separated  Configuration}})\label{defn:Separated-Config}
     Let \(\mathcal{S}= \{s_1, s_2, \ldots, s_{k} \}\) be a set of $k$ concentric semicircles whose diameter is perpendicular to the $y$-axis. The radius of the semicircle \(s_i\) is \(rad_i=i\cdot rad\), where $1 \leq i \leq k$ and $rad$ is the radius of the innermost semicircle. A separated configuration for a set of robots $R = \{r_1, r_2, \cdots, r_n\}$ with the colors from the set $C= \{c_1, c_2, \cdots, c_k\}$ is the configuration in which the robots satisfy the following predicate:
\begin{multline}
    \texttt{SepSC} = \Bigg\{ \exists t : (\forall t' > t,\, P(t) = P(t')) \text{ and } (\forall c_i \in C,\, 
\exists s_i \in \mathcal{S} \text{ such that } \\
     \left( \forall (p_j(t), c_j) \in P(t),\, c_i = c_j \iff p_j(t) \in s_i \right)) \text{ and } \\
    (\forall c_i, c_{i'} \in C,\, c_i \neq c_i' \Rightarrow s_i \neq s_{i'}) \Bigg\}.
\end{multline}
\end{definition}

 Here $c_i$ is a color from the color set $C$ and $c_j$ represents the color of the robot $r_j \in R$. In other words, the configuration is considered as separated if all robots are located on concentric semicircles whose diameters are perpendicular to  \(y\)-axis, with radii $rad, 2\cdot rad, 3\cdot rad, \dots, k\cdot rad$, where $rad$ is the radius of the smallest semicircle and $k$ is the total number of colors. All robots of same color occupy same semicircle, and each semicircle contains robots of only one color.
\begin{definition} ({\textbf{Problem Definition}})
    Given $n$ unconscious, opaque, and oblivious silent point robots, with no knowledge of $n$ but agreeing on one axis, deployed in an arbitrary initial configuration on the Euclidean plane, the problem aims to design an algorithm for the robots to reach a \texttt{separated configuration}.
    
\end{definition}

\subsection{Notations and Terminologies:}
\label{sec:notation&Terminology}
\begin{itemize}
    \item $L_r$ is the line passing through $r$ and perpendicular to $y$-axis (an horizontal line through $r$). $L_r^{\perp}$ is the line passing through $r$ and parallel to $y$-axis (an vertical line through $r$).

    \item The line segment joining the two robots $r$ and $r'$ is denoted by $\overline{rr'}$, and the line passing through those two robots is denoted by $\overleftrightarrow{rr'}$.
    
    \item  $P(t)$ denotes the configuration of robots at time $t$.
    
    \item The \textit{smallest enclosing rectangle} of the configuration $P(t)$, denoted by $\delta(P(t))$, is the rectangle with the minimum area that has two sides parallel to the $y$-axis and contains all robots at time $t$, either inside or on its boundary.

    \item The sides of the rectangle \( \delta(P(t)) \) that are parallel to the \( y \)-axis are referred to as the \textit{vertical sides}, and those perpendicular to the \( y \)-axis are referred to as the \textit{horizontal sides}. Moreover, we can distinguish between the \textit{top} and \textit{bottom} sides of the rectangle, as the robots share an agreement on the \( y \)-axis. However, the left and right sides cannot be distinguished.

    \item The robots located on the vertical sides of $\delta(P(t))$ are termed as \emph{terminal robots}. Note that each vertical side of the rectangle $\delta(P(t))$ contains at least one terminal robot. Since the robots have the agreement on the $y$-axis, a robot $r$ can identify itself as the terminal robot by checking the following: if one of the open half-planes delimited by $L_r^{\perp}$ contains no robots, then $r$ is the terminal robot; otherwise, it is non-terminal.
    
    \item The \textit{diameter robots} are the vertically lowest terminal robots on each vertical side of $\delta(P(t))$. At any time $t$, there are exactly two diameter robots.

    \item \textbf{Triangular Configuration:} A configuration $P(t)$ is referred to as a \emph{triangular configuration} if it satisfies the following conditions:

        \noindent{1.} The bottom side of the enclosing rectangle $\delta(P(t))$ contains exactly two robots located at its endpoints.

        \noindent{2.} All robots lie either inside or on the boundary of an isosceles right-angled triangle whose hypotenuse aligns with the bottom side of $\delta(P(t))$.

\end{itemize}

\section{Algorithm to Separate Opaque Colored Robots}
In this section, we present an algorithm that solves the separation problem on concentric semicircles with opaque and unconscious robots. We begin with a high-level overview of the algorithm, followed by its detailed description.

\subsubsection{A high-level Idea of the Algorithm:}
\label{sec-high-level-algo}

In this section, we give a high-level idea of our proposed algorithm named \textsc{Con-SemCirc-Separation}. The algorithm is divided into multiple stages. Each stage is treated as a subroutine.
In the first stage of the algorithm (\textsc{Arbitrary-To-Triangular}, Section \ref{subsec:arbitrayToTriangular}), the robots will move from an arbitrary configuration to a triangular configuration defined earlier. The idea is to enclose all the robots in a right-angled isosceles triangle whose hypotenuse is defined as the line segment joining the two diameter robots.
In the next stage (\textsc{Triangular-To-Semicircular}, Section \ref{subsec:TriangularToSemicircle}), the robots will project themselves on a semicircle whose diameter is formed by joining the two diameter robots in an ordered manner. Once all the robots project themselves on a single semicircle, then the robots will be mutually visible to each other and can calculate the total number of robots $n$. 


Thereafter, each robot divides the semicircle into several grid points and relocates itself on some of those grid points (\textsc{Semicircular-To-GridPoint}, Section \ref{subsec:SemiciricleToGridPoint}).
In the subsequent stage (\textsc{GridPoint-To-SectorPartitioning}, Section \ref{subsec:gridpointToSector}), all robots recognize leader(s) on the semicircle, who initiate signalling for the rest of the robots.
Although silent, the leader robot signals the rest of the robots by moving to specific points on the semicircle, each encoded as an ordered pair representing a position of a particular robot and its color.
Non-leader robots decode the leader's position to determine their respective colors and destination points, ensuring that robots of the same color occupy the same sector (an arc that subtends a fixed angle at the center) of the semicircle.
Once the robots are separated by the colors, in the final stage (\textsc{SectorPartitioning-To-ConcenSemiCirc}, Section \ref{subsec:sectorToSeparartion}), they move to their designated semicircle based on the predefined order. Hence, completing the separation.


\section*{Description of Algorithm: \textsc{Con-SemCirc-Separation}}
\label{sec:description-algo}
Here, we describe the algorithm in detail. The algorithm is divided into multiple stages.
Since the robots are oblivious, they cannot remember their previous positions and actions.
By seeing its surroundings (local view), each robot selects the appropriate stage (a subroutine) and executes the corresponding action.
Due to obstructed visibility, robots lack knowledge of the global configuration, and hence their selected stage may not align with the global one.
As a result, some robots may execute different stages simultaneously, creating overlaps.
However, our analysis shows that, despite such overlaps, all robots within finite epochs, agree on a common stage before moving to the next.
The algorithm is presented from the viewpoint of a robot $r$, whose actions are determined by its current characterization.
We state corresponding lemma at the end of each stage. The proof of each are deferred at the end of the description of all stages (Section \ref{app:analysis}).
We begin with the following stage.
Fig. \ref{fig:Stage-Representation} represents the use of the stages, to inevitably change the robots configuration.
 \begin{figure}
     \centering
     \includegraphics[width=0.8\linewidth]{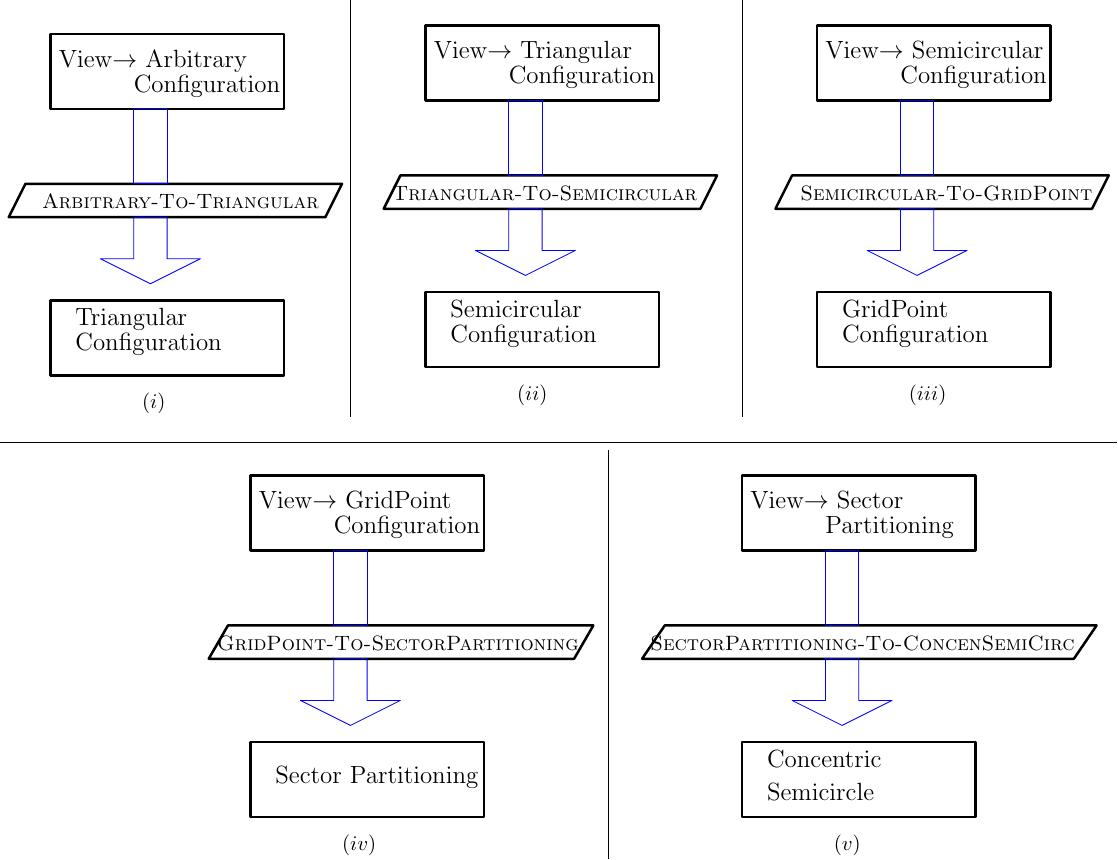}
     \caption{Represents the eventual configuration to be attained, while executing the respective subroutines, based on some specific view.}
     \label{fig:Stage-Representation}
 \end{figure}

\subsection{Stage: \textsc{Arbitrary-To-Triangular}}
\label{subsec:arbitrayToTriangular}
Our first step is to transform the initial arbitrary configuration into a triangular configuration.
To do this, we move the diameter robots or the robots that are on the bottom line of the enclosing triangle. Thus, we divide the stage into the following cases depending on the position of $r$.

 \noindent \textbf{Case A ($r$ is a diameter robot):}
 It first computes the minimum angle between the line \(L_r^{\perp}\) and the lines connecting it to any visible robot located above it, as shown in Fig. \ref{fig:minimum_angle}.
 Let $\theta$ be the minimum angle, and $r_a$ be the corresponding robot that makes this angle with $r$. 
 If $\theta < \pi/4$, the robot $r$ considers the open half-plane $\mathcal{H}_r$ delimited by the line $L_r^{\perp}$ that contains no robots. It then finds a target point $t_r$ on $\mathcal{H}_r \cap L_r$ such that the angle between the lines $\overleftrightarrow{r_at_r}$ and $L_r$ is $\pi/4$. Finally, $r$ moves to the point $t_r$.
 A diameter robot $r$ may need to repeat these steps multiple times, as there could be a robot $r''$ behind $r'$ on the line $\overleftrightarrow{rr'}$ and $r$ can not see $r''$ at the current LCM cycle due to its obstructed visibility. In such a case, $r$ needs to move further until the minimum angle $\theta$ becomes exactly $\pi/4$. 

\begin{figure}[htbp]
    \centering
    
    \begin{minipage}{0.46\textwidth}
        \centering
        \includegraphics[width=1.03\linewidth]{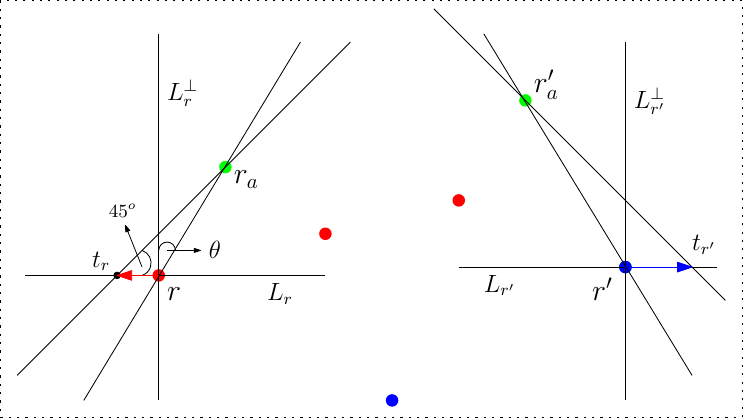}
        \caption{The diameter robots $r$ and $r'$ find their respective $\theta < \pi/4$, and move to their target point.}
        \label{fig:minimum_angle}
    \end{minipage}
    \hfill
    \begin{minipage}{0.52\textwidth}
        \centering
        \includegraphics[width=1.0\linewidth]{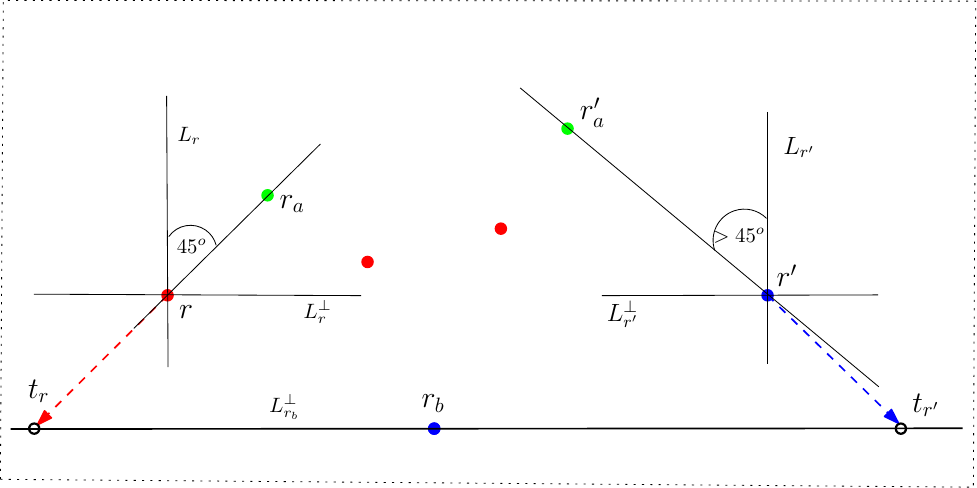}
        \caption{The diameter robots $r$ and $r'$ find their respective $\theta \geq \pi/4$, and move to the target point on $L_{r_b}$.}
     \label{fig:diameter_robot_bottom_line}
    \end{minipage}
\end{figure}

If \(\theta \geq \pi/4 \), $r$ checks if there is any robot below it. If there is no such robot, it remains in place. Otherwise, it finds $r_b$, the lowest visible robot below it. Thereafter, it finds the target point $t_r$ on $\mathcal{H}_r \cap L_{r_b}$ such that the angle between the lines $\overleftrightarrow{rt_r}$ and $L_{r_b}$ is $\pi/4$, and then finally moves to it, as shown in Fig. \ref{fig:diameter_robot_bottom_line}.


\noindent \textbf{Case B ($r$ is not a diameter robot, but on the bottom side of $\delta(P(t))$:}\\
If $r$ finds at most one robot on the bottom side of the enclosing rectangle $\delta(P(t))$ other than itself, it remains status quo.
On the other hand, if it sees two other robots on the bottom side of $\delta(P(t))$, it computes $v$ and $h$, where $v$ (resp. $h$) is the minimum vertical (resp. horizontal) positive distance from $r$ to any other robot above it.
Then it calculates a target point $t_r$, which is vertically above $r$, and the distance between $r$ and $t_r$ is $v$.
If the target point is unoccupied, and it moves to \(t_r\). Whereas if the point \(t_r\) is already occupied, \(r\) calculates the new target point \(t_r'\) that is a horizontally \(h/3\) distance and vertically \(v\) distance away from \(r\).
The measure of $h/3$ is chosen to ensure collision-freeness of the algorithm, since another robot from $L_r$, located at a horizontal distance $h$ from $r$, may also target to move synchronously a distance $v$ upward.
Finally, the robot $r$ moves to $t_r'$, as depicted in Fig. \ref{fig:Bottom_line_non-diameter_moves}.

After this stage, we ensure that all the robots reach the triangular configuration from the initial one. The corresponding lemma supporting this claim is stated below. The proof is deferred to the analysis section (Section \ref{app:analysis}).

\begin{restatable}{lem}{initialtotriangular}
    \label{lemma:L1}
   Starting from an arbitrary initial configuration, all the robots reach a triangular configuration in $O(n)$ epochs using stage \textsc{Arbitrary-To-Triangular}.
\end{restatable}



    \begin{figure}[h]
    \centering

    \begin{subfigure}[b]{0.3\linewidth}
        \centering
        \includegraphics[width=\linewidth]{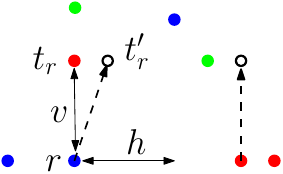}
        \caption{$r$ being non-terminal on the bottom line of $\delta(P(t))$, moves to the line of nearest vertical robot.}
        \label{fig:Bottom_line_non-diameter_moves}
    \end{subfigure}
    \hfill
    \begin{subfigure}[b]{0.3\linewidth}
        \includegraphics[width=\linewidth]{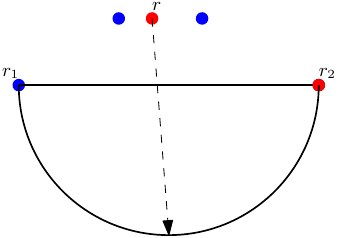}
        \caption{Case C.1.1: The semicircle is empty, and $r$ being unique nearest robot to the center, moves there.}
        \label{fig:Case.1.1}
    \end{subfigure}
    \hfill
    \begin{subfigure}[b]{0.3\linewidth}
        \includegraphics[width=\linewidth]{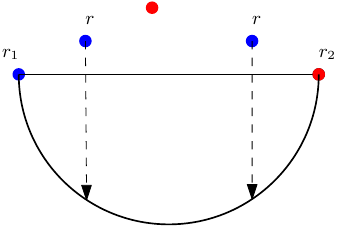}
        \caption{Case C.1.2: Two nearest robots to the center of the empty semicircle move vertically below.}
        \label{fig:Case.1.2}
    \end{subfigure}
    
    \begin{subfigure}[b]{0.3\linewidth}
        \includegraphics[width=\linewidth]{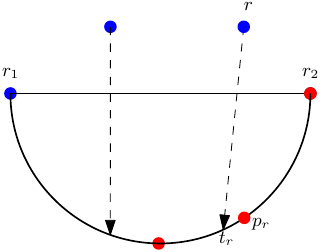}
        \caption{Case C.2: Finding a robot at $low(\lsc$ $(r_1, r_2))$, $r$ moves to a free point on the semicircle.}
        \label{fig:Case.2}
    \end{subfigure}
    \hfill
    \begin{subfigure}[b]{0.3\linewidth}
        \includegraphics[width=\linewidth]{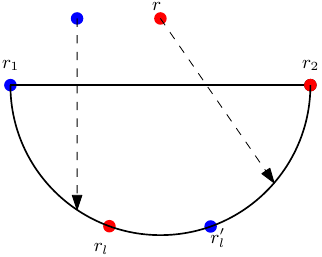}
        \caption{Case C.3: $r$ moves to the midpoint of the arc of the semicircle joining $r_l'$ and $r_2 (=r'')$.}
        \label{fig:Case.3}
    \end{subfigure}
    \hfill
    \begin{subfigure}[b]{0.3\linewidth}
        \includegraphics[width=\linewidth]{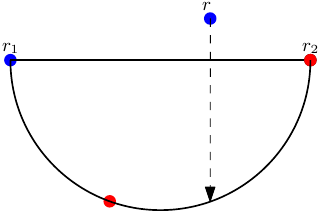}
        \caption{Case C.4: $r$ moves vertically below after finding a robot on the semicircle, equidistant from center.}
        \label{fig:Case.4}
    \end{subfigure}
    \caption{Illustration of Case B, first stage and all the sub-cases of the second stage.}
    \label{fig:all_cases}
\end{figure}
\subsection{Stage: \textsc{Triangular-To-Semicircular}}
\label{subsec:TriangularToSemicircle}
After reaching the triangular configuration, the next objective is to transform it into a semi-circular configuration. To achieve this, the robots are sequentially moved onto the lower semicircle, whose diameter is defined by the line joining the diameter robots, while keeping the diameter robots fixed in place.
In this stage, a robot $r$ first identifies two distinct robots \(r_1\) and \(r_2\), which satisfy a specific set of conditions outlined below.
\begin{enumerate}
     \item 
     Both the robots $r_1$ and $r_2$ lie on the same horizontal line, i.e., $L_{r_1} = L_{r_2}$, and the line segment $\overline{r_1r_2}$ contain no other robot besides $r_1$ and $r_2$.
    \item 
    There is no robot lying in the open region between the lines \( \overleftrightarrow{r_1r_2} \) and \( L_r \), with \( L_r \) lying above \( \overleftrightarrow{r_1r_2} \).
    \item 
    Any robot \(r_3\) positioned below the line $\overleftrightarrow{r_1r_2}$ must form a right angled triangle with \(r_1\) and \(r_2\), where the right angle is at $r_3$. 
    \item All visible robots above $\overleftrightarrow{r_1r_2}$ must lie within the interior or on the boundary of an isosceles right-angled triangle, where the hypotenuse is $\overline{r_1r_2}$.
\end{enumerate}


 Upon identifying such robots, $r$ proceeds according to the following cases.
We denote the lower semicircle having the diameter $\overline{r_1r_2}$ by $\lsc(r_1,r_2)$.
Let $cen(\lsc(r_1, r_2))$ be the centre of the semicircle $\lsc(r_1, r_2)$ and $low(\lsc(r_1, r_2))$ be the lowest point on this semicircle, which is essentially the intersection of the semicircle and the vertical line through $cen(\lsc(r_1, r_2))$.
\begin{itemize}
\item \textbf{Case C.1 (There is no robot on the semicircle $\lsc(r_1, r_2)$):}
In this case, $r$ considers all visible robots on the line $L_r$ and checks whether it is the nearest robot to $cen(\lsc(r_1, r_2))$. 
If $r$ is the unique robot nearest to $cen(\lsc(r_1, r_2))$, it moves to $low(\lsc(r_1, r_2))$, the lowest point on the semicircle $\lsc(r_1, r_2)$ (see Fig. \ref{fig:Case.1.1}). 
On the other hand, if there exists another robot $r'$ such that both $r$ and $r'$ are equidistant from the center of the semicircle, $r$ moves on the semicircle vertically below, as shown in Fig. \ref{fig:Case.1.2}.


\item \textbf{Case C.2 (There is a robot on $low(\lsc(r_1, r_2))$):}
In this case, the robot $r$ computes a point $p_r$ on the semicircle that is vertically below itself.
If $p_r$ is unoccupied, $r$ moves to $p_r$.
Otherwise, $r$ computes the minimum horizontal distance $h$ among all visible robots that do not lie on $L_r^{\perp}$.
It then calculates a point $t_r$ on the semicircle $\lsc(r_1, r_2)$ such that the horizontal distance between $p_r$ and $t_r$ is $h/3$.
Since there are two such $t_r$ on the semicircle, $r$ selects one at random and moves there. Both the cases of occupied and unoccupied $p_r$ can be seen in the Fig. \ref{fig:Case.2}.

\item \textbf{Case C.3 (There are two robots $r_{l}$ and $r'_{l}$ lying on $\lsc(r_1, r_2)$ equidistant from $low(\lsc(r_1, r_2))$,  and no robot is present on the arc of $\lsc(r_1, r_2)$ joining $r_{l}$ and $r'_{l}$):}
In this scenario, the robot $r$ first checks if it lies in the open region between the lines $L^{\perp}_{r_{_l}}$ and $L^{\perp}_{r'_{_l}}$. 
If $r$ is not in that open region, it computes the point $p_r$, which is the point vertically below itself on the semicircle.
If there is no robot at $p_r$, it moves to that point.
Otherwise, it computes a target point $t_r$ on the semicircle, as described in the previous Case C.2, and moves to it.
In contrast, if $r$ lies in the open region between the lines $L_{r_l}$ and $L_{r'_l}$, it checks the number of visible robots on $L_r$. Based on that, it executes the following actions:

\begin{enumerate}
    \item If $r$ finds two robots on $L_r$ other than itself, it remains stationary. 
    \item If $r$ detects exactly one robot, say $r'$, on $L_r$, it considers the open half-plane $\mathcal{H}_r'$ delimited by the line $L_r^{\perp}$ that does not contain $r'$.
    Note that the half-plane $\mathcal{H}_r'$ contains exactly one of $r_l$ and $r'_l$. W.l.o.g., assume it contains $r'_l$. Then $r$ moves to the midpoint of the arc of $\lsc(r_1,r_2)$ connecting $r'_l$ and $r''$, where $r''$ is the horizontally nearest robot to $r'_l$ on $\lsc(r_1,r_2)$ within $\mathcal{H}_r'$ (see Fig. \ref{fig:Case.3}).
    \item If $r$ sees no robot on $L_r$ except itself, it randomly chooses an unoccupied point on $\lsc(r_1,r_2)$ lying above the line $\overleftrightarrow{r_lr_l'}$ and moves there.
\end{enumerate}

\item\textbf{Case C.4 (There is only one robot \(r_l\) on $\lsc(r_1,r_2)$ but not on the $low(\lsc(r_1, r_2))$) ):} In this case, the robot checks whether the horizontal distance from $low(\lsc(r_1, r_2))$ to $r_l$ is equal to the horizontal distance from $low(\lsc(r_1, r_2))$ to itself. If they are equal, it projects itself vertically onto the semicircle lying below. Otherwise remains stationary.
\end{itemize}

If the robot $r$ does not satisfy any of the above conditions, it remains in place, as it will do so within finite epochs.
During this stage, some robots may incorrectly identify different robots as the diameter robots and, instead of moving to the designated semicircle, may enter the interior.
In such a scenario, the terminal robots $r_1$ and $r_2$ again follow \textsc{Arbitrary-To-Triangular} and move horizontally downward to restore the triangular configuration. However, we prove in our analysis that this situation can occur at most once.
Thereafter, all robots inevitably reach a semicircular configuration with diameter as $\overline{r_1r_2}$, the line segment joining the two diameter robots.

\begin{restatable}{lem}{atmostonemisinterpret}
    After all the robots reach to a triangular configuration by executing \textsc{Arbitrary-To-Triangular}, only one robot may misinterpret two other robots \( r'_1 \) and \( r'_2 \) as diameter robots and project itself below the actual diameter line $\overleftrightarrow{r_1r_2}$.
    \label{lemma:L2}
\end{restatable}

\begin{restatable}{lem}{traingulartosemicircle}
    \label{lemma:semiCirc}
    From the triangular configuration, all the robots reach to a single semicircle with diameter as $\overline{r_1r_2}$ by following \textsc{Triangular-To-Semicircular}.
\end{restatable}
\subsection{Stage: \textsc{Semicircular-To-GridPoint}}
\label{subsec:SemiciricleToGridPoint}
In this stage of the algorithm, the robots settle onto grid points so that the leader robots can sequentially signal the others to relocate to their assigned parts of the semicircle.
When a robot $r$ observes that all the visible robots lie on a semicircle whose diameter is perpendicular to the $y$-axis, with two robots positioned at the endpoints of this diameter, it executes the following algorithm. We begin with the computation of grid points, followed by the movement to those points. 

\noindent \textbf{Grid-Point Computation:} First, the semicircle is divided into \(2k\) arcs of equal arc-length, where $k$ is the total number of colors known to all robots.
Each arc is called a \textit{sector}, subtending an angle $\frac{\pi}{2k}$ with the centre of the semicircle $cen(\lsc(r_1, r_2))$, where $r_1$ and $r_2$ are the diameter robots.
The sectors are indexed from the diameter $\overline{r_1r_2}$ downward toward the $low(\lsc(r_1, r_2))$, numbered sequentially from 1 to $k$.
The $j$-th sector is the $j$-th division of the semicircle from the diameter $\overline{r_1 r_2}$ toward the $low(\lsc(r_1, r_2))$, counted along each half.
Each sector from both halves of the semicircle is mapped to a unique color of $C$, with the aim that all robots sharing the same color occupy the same sector.
Note that the set \(C\) is totally ordered, and the order is known to all robots in advance.
For instance, if \(C = \{c_1, c_2, c_3\}\) with $c_1 < c_2 < c_3$, then all the robots with the color $c_i$ are aiming to assign on the $i$-th sector $(1\leq i \leq 3)$.

We further subdivide each arc into \(n\) equal smaller arcs.
This results in \( 2kn \) equal arcs in the semicircle, each subtending an angle \( \alpha = \frac{\pi}{2kn} \) with the center.
The endpoints of these arcs are termed \textit{grid points}.
Note that although the robots do not initially know $n$, they can compute it from the fact that all robots are currently positioned on the semicircle.

\noindent \textbf{Movement to Grid-Points:} Using the grid points as reference, a robot $r$ checks whether all robots are located on the grid points. If not, it computes its destination position as follows:

The grid points are indexed according to their angular positions with respect to center and the diameter. A grid point located at an angle \(i\alpha\) is indexed as \(i\).
The robot $r$ first determines the half of the semicircle in which it currently belongs and then finds its destination point within that half.
If $r$ finds itself as the \( m \)-th robot from the top in its half, i.e., there are $(m-1)$ robots vertically above it within the same half, it computes the grid point $t_r$ with index \(m-1\), as its target.
If there is no other robot on the closed arc of the semicircle segmented by its current position and $t_r$, the robot $r$ moves to the grid point $t_r$. Otherwise, it remains in place until all those robots settle into their respective target positions.


\begin{restatable}{lem}{gridpoint}
    \label{lem:gridPointConfig.}
    From the semicircular configuration, all the robots reach to a grid-point configuration in $O(n)$ epochs using \textsc{Semicircular-To-GridPoint}.
\end{restatable}


\subsection{Stage: \textsc{GridPoint-To-SectorPartitioning}}
\label{subsec:gridpointToSector}
In this stage of the algorithm, our objective is to position the robots into their designated sectors.
We divide this stage into several parts.
Once all robots are aligned on the grid points, one or two robots are designated as \textit{leader} robots (\textit{leader identification}).
If necessary, these leaders first move to specific grid points, after which the remaining robots relocate to their assigned sectors based on the positions of the leader (\textit{signalling procedure}).
After that, leader robots need to be repositioned into their sectors (\textit{leader repositioning}).
Note that the diameter robots may also need to relocate, as they might not belong to their designated sectors. 
In such cases, once they move, the original reference points of the semicircle, namely $r_1$ and $r_2$, will no longer be present. 
However, since the robots agree on the $y$-axis and $n > 3$, they can still agree on a common semicircle. 
In particular, there always exists a unique semicircle passing through them whose diameter is perpendicular to the $y$-axis. 
Therefore, we simply refer to this semicircle as $\lsc$, instead of $\lsc(r_1, r_2)$.

\noindent \textbf{Leader Identification:} There can be two cases:
\begin{itemize}
    \item If there is a robot, say $r_{\ell}$, on the lowest point on the semicircle, $low(\lsc)$, it designates itself as the leader robot.
    \item On the other hand, the two robots located on the bottommost line of each half (left and right) of the semicircle are chosen as leaders.
    Let $r_{\ell}^1$ and $r_{\ell}^2$ denote those robots.
    The leader robot $r_\ell^1$ (resp. $r_\ell^2$) moves to the grid point $p_1$ (resp. $p_2$), where $p_1$ (resp. $p_2$) is the grid point next to $low(\lsc)$ and nearest to $r_{\ell}^1$ (resp. $r_{\ell}^2$). 
\end{itemize}
Subsequently, the leader robot(s) initiate the signalling process, as described below, by moving to a designated point on the arc between $p_1$ and $p_2$, referred to as \textit{signalling arc}. 
Any robot(s) positioned on this arc is considered a leader robot(s) and can be recognized as such by all other robots.
 
\noindent \textbf{Signalling Procedure:}
We divide this procedure into two cases, depending on the number of leader robots: single or dual.

\noindent\textbf{Case 1: (A single leader robot $r_{\ell}$ exists):}
Let $len$ be the length of the arc of the semicircle that subtends an angle \( \alpha \) at the centre $cen(\lsc)$. Thus, $len = rad(\lsc) \cdot \frac{\pi}{2kn}$, where $rad(\lsc)$ is the radius of the semicircle $\lsc$.
Now we define the \textit{signalling step} size as: $\tau = \frac{len}{(nk)^2}$. Next, if the leader $r_{\ell}$ intends to signal a robot located at grid point indexed $s_1$ to move to grid point indexed $s_2$, it moves to a position on the semicircle that is $g$ signalling steps away from $low(\lsc)$, towards the signalled robot, where $g = s_1 \cdot nk + s_2$.

When a robot $r$ finds another robot on the arc between the points $p_1$ and $p_2$, it recognizes that robot as the leader robot $r_{\ell}$, where $p_1$ and $p_2$ are the grid points adjacent to $low(\lsc)$.
If both $r$ and $r_{\ell}$ lie on the same half of the semicircle, the robot $r$ computes $g$, where $g\cdot \tau$ is the arc distance between $low(\lsc)$ and $r_{\ell}$.
It then obtains $s_1$ and $s_2$ from the relations: $s_1= \left\lfloor \frac{g}{nk} \right\rfloor$ and $s_2=g\mod(nk)$.

If the robot $r$ occupies the grid point indexed $s_1$, this implies that the leader robot is signalling to $r$. In that case, $r$ moves to the grid point indexed $s_2$ in its own half of the semicircle.
If the grid point $s_2$ is already occupied, $r$ instead moves to the corresponding grid point $s_2$ in the opposite half. This situation arises only when the leader robot signals a robot to move to a diameter point(i.e., when $s_2=0$), as described below.

The leader robot $r_\ell$ signals the other robots sequentially, with a certain prioritization. It first assigns the diameter positions to the robots of color $c_k$, which is the highest-ordered color in $C$.
Such a settlement is treated as an exception, since the general objective is to gather the robots of color $c_i$ into the $i$-th sector ($1\leq i \leq k$).
This is done for the final configuration, where all the robots of the color $c_k$ are placed on the innermost semicircle with two robots of color $c_k$ occupying diameter points.
To achieve this, several cases may arise based on the occupancy of diameter points and the number of visible $c_k$-colored robots by $r_\ell$:
\begin{itemize}
    \item If the diameter points are already occupied by robots of colors other than $c_k$, then the leader robot $r_\ell$ first signals those robots to move to their designated sectors.
    \item If both the diameter points are unoccupied and $r_\ell$ finds at least two robots of color $c_k$, it arbitrarily selects two of them and signals sequentially to move to the diameter points.
    \item If both the diameter points are unoccupied and $r_\ell$ finds only one robot of color $c_k$, then $r_\ell$ itself must be of the color $c_k$, as we assume that at least two robots of each color exist. In this case, $r_\ell$ first signals the only visible robots of color $c_k$ and then all the robots that are not in their designated sectors. Thereafter, it moves to the unoccupied diameter points.
    \item If exactly one diameter point is unoccupied and $r_\ell$ finds only one robot of color $c_k$ other than that lying on the diameter point, the leader robot $r_\ell$ signals that robot to move to the unoccupied diameter point.
    \item If exactly one unoccupied diameter point exists and no visible robot of color $c_k$ other than the one already on a diameter point, then $r_\ell$ moves to that unoccupied point after all other robots are separated into sectors.
\end{itemize}

Apart from signaling robots of color $c_k$ to occupy the diameter points, the leader robot $r_{\ell}$ signals the other robots as follows. It first determines $s_1$ by identifying the topmost robot $r$ that is not currently positioned in its designated sector. 
Recall that a robot of color $c_j$, the $j$-th lowest-ordered color in $C$, is assigned to the $j$-th sector. 
The leader can determine the correct sector of $r$ from its color. 
If two such robots exist, the leader selects one at random. 
It then determines $s_2$, the index of the topmost unoccupied grid point in the $j$-th sector.

When the leader robot $r_{\ell}$ finds itself on the arc between $p_1$ and $p_2$, it computes $g$, $s_1$, and $s_2$, as described above. 
If $r_{\ell}$ observes that the grid point indexed $s_1$ is unoccupied, it restarts the signalling procedure. This process is continued until the leader finds all the robots placed in their respective sectors, with two robots of color $c_k$ (if they exist) occupying the diameter points.
Afterwards, if both the diameter points are occupied with the robots of color $c_k$, $r_\ell$ moves to $low(\lsc)$. However, if one of them remains unoccupied, the leader moves there, as discussed in the above sub-cases.




\noindent\textbf{Case 2: Two leader robots $r_{\ell}^1$ and $r_{\ell}^2$ exist.}
In this case, both leaders \( r_{\ell}^1 \) and \( r_{\ell}^2 \) initiate signalling procedure within their respective halves, analogous to the single-leader case. 
The key distinction arises when the signalling robots of color $c_k$.
Each leader first checks whether there is any $c_k$-colored robot in its own half.
If such a robot exists, it signals that robot. Otherwise, it determines the number of $c_k$-colored robots in the other half. If there are two such robots, the leader proceeds with signalling the robots in its own half.
If there is only one robot of color $c_k$ in the other half, then the leader itself must be of color $c_k$. In that case, it moves to the unoccupied diameter point after all the robots get separated into sectors.
Finally, once all robots are placed in their respective sectors, with two $c_k$-colored robots at the diameter points, the leader robots $r_\ell^1$ and $r_\ell^2$ (if they still exist) respectively move to the grid point $p_1$ and $p_2$, instead of $low(\lsc)$.

\noindent \textbf{Leader Repositioning:} 
After all the robots settle within their sectors, the leader robot(s) need to be positioned in the correct designated sector, if not correctly positioned.
All the robots except the leader can detect whether the leader robot is in the correct sector or not.
We again divide this procedure into two cases as before.

    \noindent \textbf{Case 1: (A single leader robot $r_\ell$ exists):} After the signalling procedure, the leader robot $r_\ell$ must be positioned at $low(\lsc)$. Upon seeing a robot on $low(\lsc)$ and all other robots on their designated sector, the nearest robot $r$ of color $c_j$, that matches with the leader, projects itself horizontally onto the line joining two neighbouring grid-points of its current position.
    This action signals to the leader that $r$ shares the same color.
    When the leader robot at $low(\lsc)$ observes $r$ on the line joining grid points indexed $i$ and $(i+2)$, for some $i$, while all other robots are placed on the correct sector, it moves to one of the unoccupied grid points of the sector containing the grid point indexed $(i+1)$. 
    However, if the leader $r_\ell$ instead sees another robot (other than $r$) outside its designated sector, it disregards this signal and continues the signalling process.
    Finally, whenever $r$ sees that the leader is no longer at $low(\lsc)$, it returns to its grid point indexed $(i+1)$, the nearest intersection of $L_r$ and semicircle $\lsc$.

\begin{figure}[t]
    \centering
    \begin{subfigure}[b]{0.45\linewidth}
        \centering
        \includegraphics[width=0.8\linewidth]{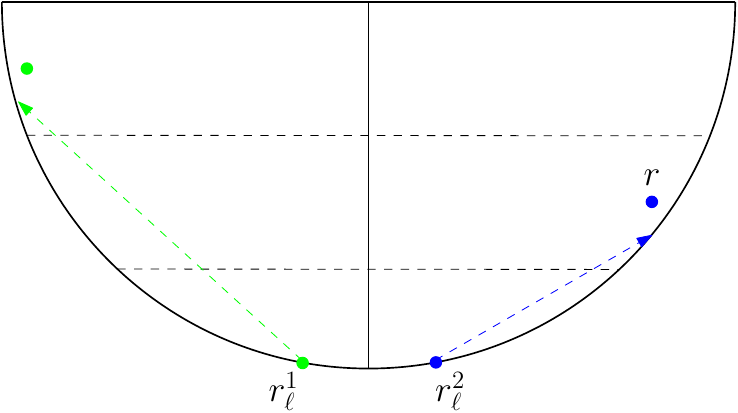}
        \caption{Robot moves horizontally on the line joining the neighbouring grid points to signal its half's leader.}
        \label{fig:Leader_repositioning_1}
    \end{subfigure}
    \hfill
    \begin{subfigure}[b]{0.45\linewidth}
        \centering
        \includegraphics[width=0.8\linewidth]{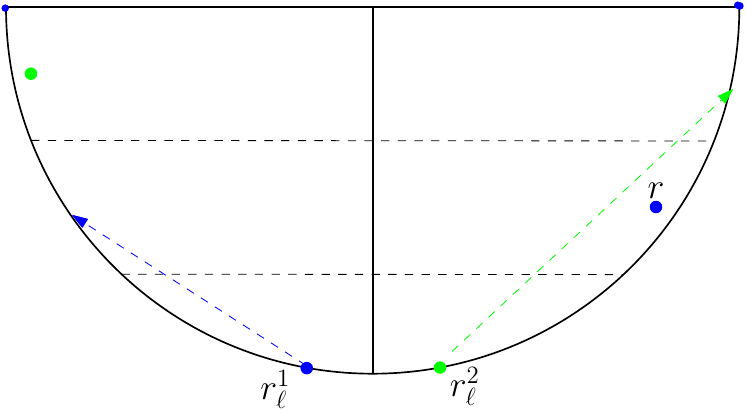}
        \caption{Robot moves the midpoint of the line joining the neighbouring grid points to signal other half's leader.}
        \label{fig:Leader_repositioning_2}
    \end{subfigure}
    \caption{Leader repositioning with different inward distances.}
    \label{fig:Leader_repositioning}
\end{figure}
    
    \noindent \textbf{Case 2: (Two leader robots $r_{\ell}^1$ and $r_\ell^2$ exist):} After the signalling procedure, the leader robots occupy positions $p_1$ and $p_2$.
    Upon seeing a robot or two robots on $p_1$ or $p_2$ or both, and all other robots on their designated sectors, the robot $r$, that is vertically nearest to a leader robot and sharing the same color, initiates signalling. 
    It does so by projecting itself either horizontally on the line joining the neighbouring grid-points or on the midpoint of that line.
    If the robot $r$ is signalling to the leader in the same half, the robot horizontally projects itself horizontally on the line joining the neighbouring gridpoints. 
    If instead it is signalling to the leader in the opposite half, it moves to the midpoints of that line.
   Depending on the position of the signalling robot, the leader robot(s) move to the unoccupied grid point of the sector, within their own half, to which the signalling robot belongs (see Fig. \ref{fig:Leader_repositioning}).
    Once the leader robot moves from $p_1$ (or $p_2$ ), which means there is no robot on $p_1$ (or $p_2$), the signalling robots return to their original grid point, which is either the nearest grid point lying on $L_r$ or the grid point in its own half.

\begin{restatable}{lem}{gridtosector}
\label{lem:sectorPart}
    From the grid-point configuration, all the robots reach to their designated sectors in $O(n)$ epochs using \textsc{GridPoint-To-SectorPartitioning} without collision. 
\end{restatable}

\subsection{Stage: \textsc{SectorPartitioning-To-ConcenSemiCirc}}
\label{subsec:sectorToSeparartion}
After all the robots are partitioned into the sectors, in this stage, the robots separate themselves into concentric semicircles.
Once all robots are separated into their assigned sectors, and no robot is in the signalling arc.
The robots that are not on the diameter points of the semicircle execute the following algorithm to separate into concentric semicircles.
The robots at the diameter points act as the reference. The line joining them is the diameter of the innermost semicircles.
A robot $r$, not on the diameter points, checks the following conditions:
\begin{itemize}
    \item All the robots, except diameter robots, are in their designated sectors.
    \item The signalling arc of the semicircle (the arc joining $p_1$ and $p_2$, excluding the endpoints and $low(\lsc)$) is empty.
\end{itemize}
If the robot $r$ finds both conditions true, it calculates its next destination. The robot knows the total number of colors  $k$. The lowest sector ($k$-th sector) of the semicircle in each half is assigned the number 1, and the topmost sectors in the semicircles are assigned the number $k$. 
In general, the $i$-th sector in each half is numbered $(k-i+1)$.
Now let $i$ be the number assigned to the sector containing $r$. The destination point $t_r$ of the robot $r$ is defined as the intersection of: the concentric semicircle with radius $i \cdot rad(\lsc)$ and line $\overleftrightarrow{rr_d}$, where $rad(\lsc)$ is the radius of the smallest semicircle $\lsc$, where $r_d$ is the diameter robot in the opposite half of the semicircle.
Finally, $r$ moves to $t_r$ only if all the robots above it, except the diameter robot, have already moved to their respective semicircle.

However, if $r$ finds that the robot $r'$ directly above it is not on the correct semicircle corresponding to its color, $r$ moves to the midpoint between two consecutive grid points on the semicircle. This temporary move by $r$ serves as a signal to $r'$ for returning to the smallest semicircle.  
Such a situation arises when $r$ mistakenly interprets itself as the robot active in the current stage, while in fact the signalling procedure has not yet completed. By deliberately occupying a non-grid point, $r$ enables the misinterpreted robot to return to the innermost semicircle. This ensures that the grid-point configuration first needs to be restored, and after which the signalling procedure can proceed correctly.

Note that, after some robots move from the innermost semicircle to the outer semicircles, the previous diameter robot may no longer remain the diameter. Recall that a diameter robot is defined as the vertically lowest terminal robot on the vertical side of the enclosing rectangle. Here, we slightly misuse this definition: a robot $r$ considers two robots $r_d$ and $r'_d$ as the diameter robots if all other robots either (i) lie on the correct sectors of the semicircle $\lsc(r_d, r'_d)$, or (ii) are on their designated concentric semicircle with the innermost semicircle defined by $\lsc(r_d, r'_d)$.
At the end of this stage, robots are separated into concentric semicircles, as argued in Lemma \ref{lem:finalSeparation}. Fig. \ref{fig:final-config} represents the grid point to the respective semicircle movement, and attaining the final configuration, where red robots and green robots are the highest and lowest ordered colored robots. 

 \begin{figure}[t]
    \centering
    \begin{subfigure}[b]{0.45\linewidth}
        \centering

     \includegraphics[width=0.9\linewidth]{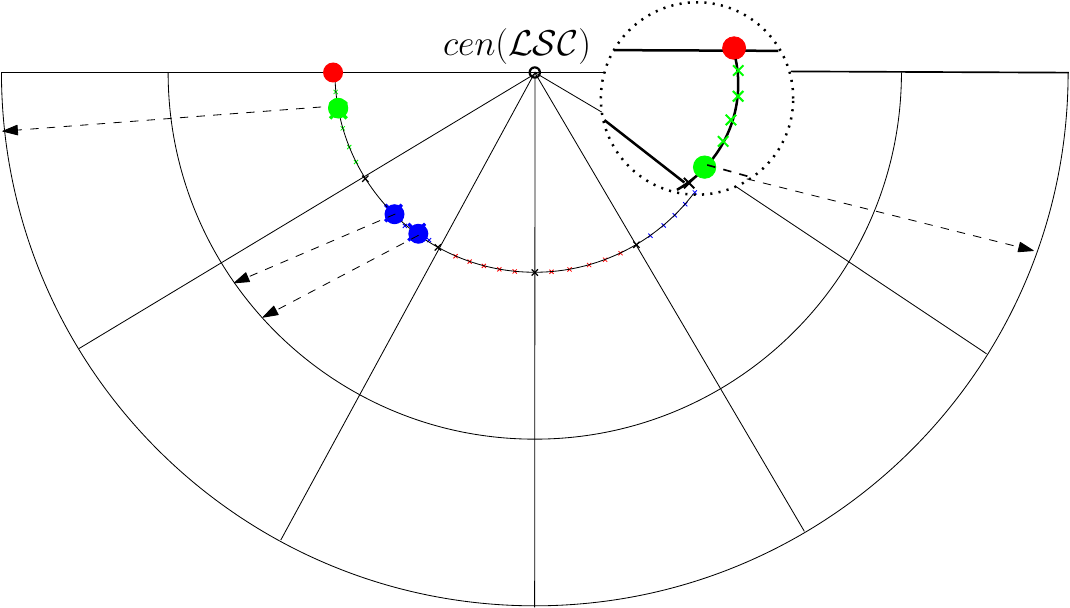}
     \caption{Illustrates the movement of the robots from their gridpoint position, to their respective semicircle}
     \label{fig:gridpointmovement}
 \end{subfigure}
    \hfill
    \begin{subfigure}[b]{0.45\linewidth}
        \centering
        
\includegraphics[width=0.9\linewidth]{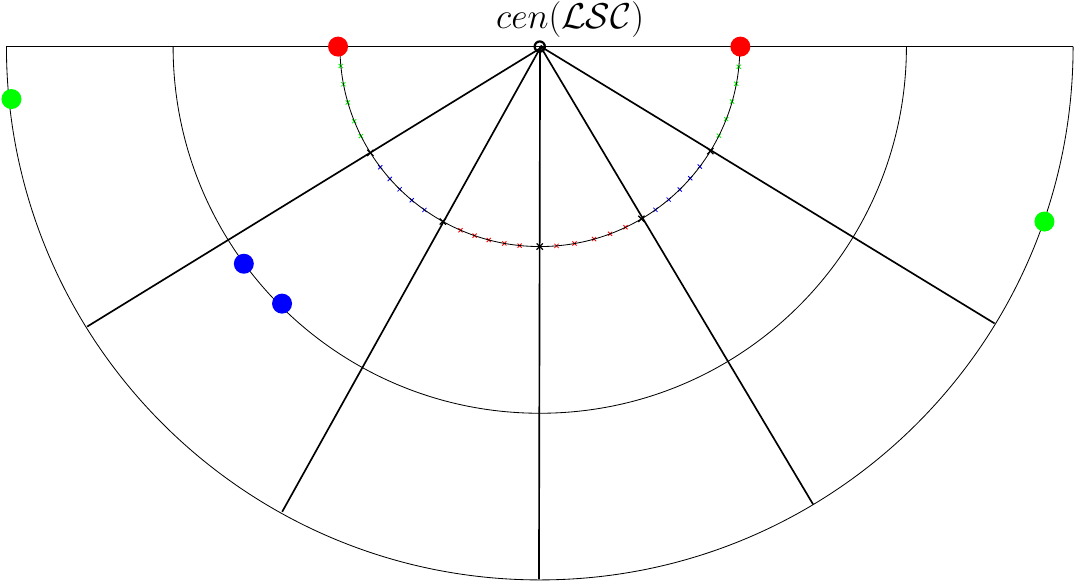}
     \caption{Represents the final configuration, where the robots are in their respective semicircle}
     \label{fig:finalconfiguration}
 \end{subfigure}
 \caption{Movement of robots to final configuration}
 \label{fig:final-config}
 \end{figure}
 
\begin{restatable}{lem}{separation}
    \label{lem:finalSeparation}
    After the robots partitioned into sectors, all the robots proceed to their respective semicircle for the final separation using \textsc{SectorPartitioning-To-ConcenSemiCirc} in $O(n)$ epochs without collision.
\end{restatable}

\subsection{Analysis of the Algorithm \textsc{Con-SemCirc-Separation}}
\label{app:analysis}
Here, we establish the correctness of the algorithm in achieving collision-free separation of robots into concentric semicircles according to their colors, along with the time complexity analysis.

\initialtotriangular*

\begin{proof}
    Let $r_1$ be a diameter robot. By definition, one of the half-planes, say $\mathcal{H}_{r_1}$ delimited by the line \(L_r^{\perp}\) must be empty. Let $\tilde{L}_{r_1}$ be a line passing through $r_1$ making an angle of \(\pi/4\) with positive \(y\)-axis towards the half plane $\mathcal{H}_{r_1}$. The $\tilde{L}_{r_1}$ divides the Euclidean plane into two half planes:
    \begin{enumerate}
        \item the open half plane above $\tilde{L}_{r_1}$, denoted as $\widetilde{\mathcal{H}}^{up}_{r_1}$.
        \item the closed half plane below $\tilde{L}_{r_1}$, denoted as $\widetilde{\mathcal{H}}^{down}_{r_1}$.
    \end{enumerate}
    \textit{Claim 1:} \textit{Starting from the initial configuration, within $O(n)$ epochs, the diameter robot $r_1$ reaches a point such that all robots lie in $\widetilde{\mathcal{H}}^{down}_{r_1}$.}\\
    \textit{Proof:} According to the algorithm, the diameter robot \(r_1\) only moves to the empty half $\mathcal{H}_{r_1}$ horizontally or moves downwards along the line $\tilde{L}_{r_1}$.
    All the non-diameter robots that are not on the bottom line of $\delta(P(t))$ either remain stationary or move downwards by executing \textsc{Triangular-To-Semicircular}.
    So they never cross the line $\tilde{L}_{r_1}$.
    Whereas, a non-diameter robot on the bottom line of $\delta(P(t))$ may cross the line $\tilde{L}_{r_1}$, but they all move to the same horizontal line, which is vertically above themselves.
   Even in that case, to include all such robots in $\widetilde{\mathcal{H}}_{r_1}^{down}$, the diameter robot needs to move horizontally once, in constant epoch.
    If the robot $r_1$ sees any robot \(r'\) above $\tilde{L}_{r_1}$, it moves horizontally to a point such that, after the movement the robot \(r'\) is on the line $\tilde{L}_{r_1}$. Hence in each epoch, at least one robot is included in the region $\widetilde{\mathcal{H}}_{r_1}^{down}$. Hence we can say that \(r_1\) will reach a point where all the other robots are below $\tilde{L}_{r_1}$ or on the line itself in at most $O(n)$ epochs, where \(n\) indicates the total number of robots.

    We now prove that the movements of the diameter robot $r_1$ not only places each robot within $\widetilde{\mathcal{H}}_{r_1}^{down}$, but it also ensure that they lie above the line joining the robots, say $r_1$ and $r_2$, where $r_2$ is another diameter robot. To show this, we argue that the bottom-most horizontal line can change at most once. Hence, once the diameter robots reach this line, the configuration either becomes triangular, or the diameter robots must move once more due to the downward movement of a non-diameter robot, after which the configuration becomes triangular. This establishes that once all robots are within $\widetilde{\mathcal{H}_{r_1}^{down}} ~\cap ~\widetilde{\mathcal{H}}_{r_2}^{down}$, the diameter robots require only a constant number of additional epochs to reach a triangular configuration.

    \noindent \textit{Claim 2: The bottom-most line can change at most once.}\\
    \textit{Proof:} Notice that if a robot sees more than two robots below itself, it does not move unless all those robots lie on a single semicircle. Even if in that case, the robot projects but does not cross the lowest robot of the configuration.
    Thus, two cases remain: the bottom-most horizontal line of the current configuration contains either one or two robots.
    For a robot $r'$ to cross the lowest robot, it must see two robots $r_1$ and $r_2$ below it with the same $y$-coordinate, causing $r'$ to assume they are diameter robots. It then incorrectly initiates the phase \textsc{Triangular-To-Semicircular} and moves to the semicircle $\lsc(r_1, r_2)$, which creates a new bottom-most line (see Fig. \ref{fig:misinterpretation}). We argue that such a misinterpretation can occur at most once.

    \noindent\textit{Case 1:} Suppose two robots \(r_1\) and \(r_2\) lie on the lowest line of the configuration. Any robot outside the region between the lines $L_{r_1}^{\perp}$ and $L_{r_2}^{\perp}$ cannot cross \(\overleftrightarrow{r_1r_2}\), since it must see at least three robots: two obstructing \(r_1\) and \(r_2\), and one misidentified as a diameter robot by itself. 
    If a robot \(r'\) is inside the region between the lines $L_{r_1}^{\perp}$ and $L_{r_2}^{\perp}$, it may cross \(\overleftrightarrow{r_1r_2}\). However, once \(r'\) crosses, any remaining robot (inside or outside this region and above \(\overleftrightarrow{r_1r_2}\)) always sees at least three robots and therefore never crosses it. 
    Finally, the diameter robots \(r_1\) and \(r_2\) move to the line \(L_{r'}\), implying that the bottom-most line can change at most once.


\begin{figure}[ht]
    \centering
    \begin{subfigure}[t]{0.46\linewidth}
        \centering
        \includegraphics[width=\linewidth]{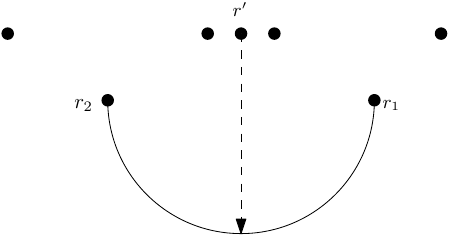}
        \caption{\(r\) projects itself on the semicircle since it can see two robots below it which share the same $y$-coordinate.}
        \label{fig:mis1}
    \end{subfigure}
    \hfill
    \begin{subfigure}[t]{0.46\linewidth}
        \centering
        \includegraphics[width=\linewidth]{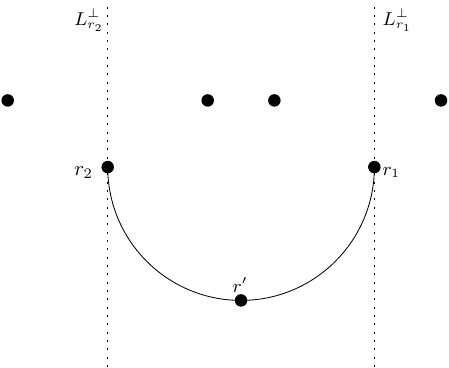}
        \caption{Any robot above \(\overline{r_1r_2}\) can see at least three robots below it, and so can not cross the bottom-most line.}
        \label{fig:mis1result}
    \end{subfigure}
    \caption{Illustration of robot visibility and projection conditions.}
    \label{fig:misinterpretation}
\end{figure}

    \noindent\textit{Case 2:} There is only one robot $r'$ on the bottom most line of the configuration. If another robot \(r\) crosses \(L_{r'}\) by misinterpreting two robots \(r_1\) and \(r_2\) as diameter robots, then after this move any robot above \(\overleftrightarrow{r_1r_2}\) sees at least three robots below and thus never crosses the new bottom line \(L_r\). 
    Any robot \(r''\) below \(\overleftrightarrow{r_1r_2}\) must lie on either \(\overleftrightarrow{r^{prev}r_1}\) or \(\overleftrightarrow{r^{prev}r_2}\), where \(r^{prev}\) is the previous position of \(r\) (see Fig. \ref{fig:misinterpretation2}). 
    Such an \(r''\) sees only one robot on \(L_r\) and hence cannot cross it. 
    Therefore, the bottom-most line can change at most once due to the movement of \(r\).

    Note that even if any one or both the diameter robots $r_1$ and $r_2 $are inactive, the phase of that robot does not change. It will always follow \textsc{Triangular-To-Semicircular}. Therefore, the semi-synchronous setting will not have any effect on the proofs given above in Claim-1 and Claim-2

    Hence, from Claim-1 and Claim-2, we can conclude that, in $O(n)$ epochs, all the robots lie above the line $\overleftrightarrow{r_1r_2}$ and inside the region $\widetilde{\mathcal{H}}_{r_1}^{down}$ ~$\cap$ ~ $\widetilde{\mathcal{H}}_{r_2}^{down}$, which essentially transforms the initial configuration to the triangular one. \qed



 \end{proof}

\begin{figure}[ht]
    \centering
    \begin{subfigure}[t]{0.45\linewidth}
        \centering
        \includegraphics[width=\linewidth]{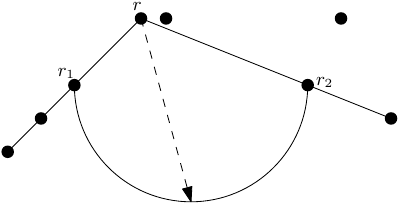}
        \caption{Robot \(r\) sees only two robots below itself. Hence, the robot \(r\) projects itself on \(\lsc(r_1, r_2)\).}
        \label{fig:mis2}
    \end{subfigure}%
    \hfill
    \hfill
    \begin{subfigure}[t]{0.45\linewidth}
        \centering
        \includegraphics[width=\linewidth]{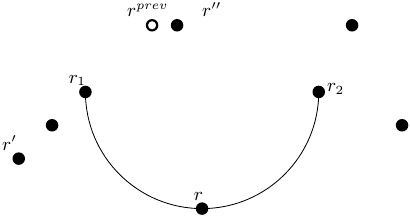}
        \caption{
        $r''$ can not cross the bottom most line $L_r$, as it is unable to see two robots on $L_r$ with the same $y$-coordinate.
        }
        \label{fig:mis2result}
    \end{subfigure}
    \caption{Robot visibility conditions for projection on \(\lsc(r_1, r_2)\).}
    \label{fig:misinterpretation2}
\end{figure}

\atmostonemisinterpret*

\begin{proof} 

Once the triangular configuration is reached, the bottom-most line of the configuration will contain exactly two robots, which are diameter robots, say \(r_1\) and \(r_2\).

Suppose there is a robot $r$ which can see two robots \(r'_1\) and \(r_2'\) satisfying all the conditions (mentioned in stage\textsc{Triangular-To-Semicircular}) to project itself on the $\lsc(r_1',r_2')$.
We know that any robot, seeing more than two robots below itself, cannot cross the bottom-most line of the configuration which is in this case, $\overleftrightarrow{r_1r_2}$.
Now suppose that \(r\) can see only two robots below itself \(r_1'\) and \(r_2'\). 
In such a case, there should not exist any robot between lines \(L_r\) and $\overleftrightarrow{r_1'r_2'}$.
Moreover, all the robots below $\overleftrightarrow{r_1'r_2'}$ are either on the line \(\overleftrightarrow{rr_2'}\) or on the line \(\overleftrightarrow{rr_1'}\). 
Hence, without loss of generality, we can say that the robots \(r_1\) and \(r_2\) are also on the lines \(\overleftrightarrow{rr_1'}\) and \(\overleftrightarrow{rr_2'}\), respectively. 
We will show that there does not exist a robot \(r'\) other than $r$ that can see only two robots except \(r_1\) and \(r_2\), which proves the claim.
The robot \(r'\) can not be above the line \(L_r\), since any robot above \(L_r\) can see at least three robots below itself. 
All the robots below \(L_r\) are located either on the line \(\overleftrightarrow{rr_2'}\) or on the line \(\overleftrightarrow{rr_1'}\). 
If the robot is located on the line \(\overleftrightarrow{rr_2'}\) then it can see robot \(r_1\) and if it is located on the line \(\overleftrightarrow{rr_1'}\), it can see the robot \(r_2\).
There do not exist two robots that can be misinterpreted as diameter robots by $r'$.
The robot \(r'\) also can not be on the line \(L_r\) since it can always see \(r_1'\), \(r_2'\) and at least one of the robot among \(r_1\) and \(r_2\). \qed
\end{proof}

\traingulartosemicircle*

\begin{proof}
 We divide the proof into two cases based on whether any robot misinterprets the diameter robots $r_1$ and $r_2$ or not.
 
\noindent \textit{Case 1: No robot misinterpret the diameter robots and project itself below the bottom most line of the current configuration.}\\
If the robots are not in the semicircular configuration, at any LCM cycle, there always exist at least one robot $r'$, lying on just above the line $\overleftrightarrow{r_1r_2}$, which can see both the diameter robots \(r_1\) and \(r_2\).
The robot $r'$ project itself on the $\lsc(r_1, r_2)$.
Once we reach the configuration where there is at least one robot on the semicircle (including $r'$) and all the other robots that are not on the semicircle are above the line $\overleftrightarrow{r_1r_2}$ and inside or on the isosceles right angled triangle with hypotenuse is $\overline{r_1r_2}$, all the robots above the line $\overleftrightarrow{r_1r_2}$ can always see at least three robots below itself.
Hence, they do not cross the line $\overleftrightarrow{r_1r_2}$, even if they misinterpret the diameter robots.
They only cross the line when they find the original diameter robots.
Since, in every epoch, there exists at least one robot that can see both the diameter robots, at least one of them must move to the semicircle in that epoch.
Thus, all the robots will reach the semicircular configuration in \(O(n)\) epochs.

\noindent \textit{Case 2: There is a misinterpretation and a robot \(r\) projects itself below the line $\overleftrightarrow{r_1r_2}$ where \(r_1\) and \(r_2\) are the diameter robots.}\\
As shown in Lemma \ref{lemma:L2}, at any given LCM cycle, only one robot can misinterpret the diameter robots and project itself below the bottom-most line of the configuration. Let us assume \(r\) is the robot which misinterprets the robots \(r_1'\) and \(r_2'\) as diameter robots. All the robots below the line \(L_r\) are located either on the line \(\overleftrightarrow{rr_1'}\) or on the line \(\overleftrightarrow{rr_2'}\).
Hence, at any given point of time, at most two robots can cross the line \(\overleftrightarrow{r_1r_2}\) along with $r$. 
Therefore, there can be two sub-cases depending on the number of robots that are simultaneously crossing the line $\overleftrightarrow{r_1r_2}$ along with $r$. 
\newline

\begin{itemize}

\item \noindent\textit{Case 2.1: Two robots cross the line \(\overleftrightarrow{r_1r_2}\) in sync. $r$}: 
If three robots cross the line \(\overleftrightarrow{r_1r_2}\), then two of them will share the same $y$-coordinate.
In the next epoch, the diameter robots $r_1$ and $r_2$ reach to the bottom-most line of the current configuration.
After that, the robot(s) on the bottom-most line, except the diameter robots, move upward \(v\) vertical distance away from themselves, where $v$ is the smallest vertical distance from any other robot.
Note that, in this cycle, no other robot crosses the bottom-most line, as that line contains three robots.
Therefore, all three robots that had crossed the line \(\overleftrightarrow{r_1r_2}\) before are now on the same horizontal line. Thereafter, no robot will misinterpret the diameter robots and project itself below the bottom-most line.
Hence, by a similar argument as discussed in Case 1, the robots will reach the semicircular configuration in \(O(n)\) epochs.\newline

\item \noindent\textit{Case 2.2: Only two robots \(r\) and \(r'\) cross the line \(\overleftrightarrow{r_1r_2}\) simultaneously: }
\(r'\) projects itself on the $\lsc(r_1, r_2)$.
If only two robots cross the line, all the robots above the line \(\overleftrightarrow{r_1'r_2'}\) do not cross the bottom-most line, as they can see at least three robots below them.
The robots below the line \(\overleftrightarrow{r_1'r_2'}\) are either on the line \(\overleftrightarrow{rr_1'}\) or on the line \(\overleftrightarrow{rr_2'}\).
After the movements of $r$ and $r'$, every robot below the line $\overleftrightarrow{r_1'r_2'}$ can see either both original diameter robots or at least one of them. Hence, such robots cannot misinterpret the diameter robots and therefore never cross the line $\overleftrightarrow{r_1r_2}$, while $r$ and $r'$ are still below $\overleftrightarrow{r_1r_2}$.
Hence, in the next epoch, the diameter robots move to the bottom-most line.
After that, $r'$ moves to the line $L_r$ in another epoch.
Again note that, before the movement of $r'$, no robot will cross the bottom-most line, as the line contains three robots.

After the movement of $r'$, any robot above the line \(\overleftrightarrow{r'r}\) finds three robots below itself.
Hence, it will not cross the bottom-most line of the configuration. But the robots \(r\) and \(r'\) can see both the diameter robot \(r_1\) and \(r_2\). Therefore, at least one of them must move on the semicircle $\lsc(r_1,r_2)$. 
Thereafter, all other robots will reach to the semicircular configuration in $O(n)$ epochs, as explained in Case 1. 
\end{itemize}
\qed
\end{proof}

\gridpoint*

\begin{proof}
  Once the robots reach the semicircular configuration by following \textsc{Triangular-To-SemiCircle} (Section \ref{subsec:TriangularToSemicircle}), they all agree on $n$, the total number of robots. 
  By the model assumption, they also agree on $k$, the total number of colors. 
  Therefore, all robots agree on the grid points, which are functions of $n$ and $k$. 

    Now, suppose there exists a robot $r$ that is not on its assigned grid point. 
    In such a case, all robots, including $r$ itself, can identify that $r$ is misplaced and compute their destination grid point, as mentioned in Section \ref{subsec:sectorToSeparartion}. 
    At any given epoch, there always exists at least one robot $r'$ such that no other robot lies on the closed arc between $r'$ and its destination $t_{r'}$. 
    If there is a robot $r''$ on that arc, then $r''$ first moves to its destination, followed by $r'$. 
    Hence, in each epoch, at least one robot settles on its assigned grid point. 
    This implies that all robots settle on some grid point in at most $O(n)$ epochs. 
    
    Notice that all the robots are located on a single semicircle. Any robot $r'$ only moves when it sees that the destination point $t_{r'}$ is unoccupied and there is no robot located on the arc of the semicircle segmented by the points $t_{r'}$ and $r'$. Notice that if there is a robot $r''$ which shares the same path as robot $r'$ then the robot $r'$ must be on the arc segmented by the robot $r''$ and its destination point $t_{r''}$. Thus, the robot $r''$ will not move. Hence in this stage, all robots settle on some grid point in at most $O(n)$ epochs without any collision.  \qed

\end{proof}

\gridtosector* 
 
\begin{proof}
    By Lemma \ref{lemma:semiCirc}, once the robots reach the semicircular configuration in $O(n)$ epochs, they agree on the numbers $n$, the total number of robots and $k$, the total number of colors.
    Thereafter, all the robots can uniquely calculate the $\tau=\frac{len}{(nk)^2}$, which is the signalling step size, where $len = \frac{rad(\lsc) \cdot \pi}{2nk}$.
    Additionally, the robots have a consensus about the leaders, as either it lies on the lowest point or its neighbouring points of the semicircle.
    Any arc segmented by the lowest grid-point and the neighbouring grid-point will contain $(nk)^2$ signalling steps. Notice that the leader robot can signal an ordered pair $(s_1,s_2)$ by reaching to a point $g$ steps away from lowest grid-point. $g$ is calculated by
    \[
        g=s_1.nk+s_2 
    \]
    which is unique for every ordered pair $(s_1,s_2)$ $0<s_1,s_2<nk$ where $nk$ is the number of the grid-points of one half of the semicircle.
    The maximum value of $g$ is $(nk-1)nk + (nk-1) < (nk)^2$, so even if the maximum number of robots are present in a half, the leader has sufficient signalling points to signal all those robots.
    All the robots can uniquely and correctly decode the signal into the ordered pair as 
    \[
        s_1= \left\lfloor\frac{g}{nk}\right\rfloor
    \]
    and 
    \[
        s_2= g~ mod (nk)
    \]
    Hence, we can say that the leader robot can signal a robot $r$, which is identifiable by the leader by its color, to go to any grid-point in the half of the semicircle. 
    The leader robot can also signal any robot to go to the diameter point on the opposite half. 
    Therefore, in at most two consecutive epochs, at least one robot settles to its designated sector, which implies that this procedure concludes in $O(n)$ epochs. 
    
    Notice that no robot moves to the signalling arc.
    Moreover, in each cycle, at most one robot from each half moves.
    We know that the chord joining any two points on one half of the semicircle will entirely lie in that half. So, the two chords from the two halves of the semicircle cannot intersect.
    Hence, if the two leader are signalling the two robots of their respective halves to move to their own halves, the line of movement of the two signalled robots can not intersect, and so they never collide.
    On the other hand, a leader signals a robot to move to the other half, only when it signals the robot to move to the diameter point, where the diameter point on its own half is occupied.
    For the leader case, both leaders can not signal such robots simultaneously. 
    Since any chord joining the diameter point of the semicircle to any point on the opposite half cannot intersect with any chord joining the two points in the same half except at the diameter point, the signalled robots can not collide in that case too.
    
    After the diameter points are fixed by the highest order colored robot and all the other robots except the leader robots reach the sector assigned, in the next three epochs, the nearest robot sharing the same color as the leader robot, signals the leader robot so it moves to the sector assigned to its color. Following the same argument as earlier, we can say that the leader robot reaches the destination point without collision.
    Hence, the robots achieve the partition into the sectors in $O(n)$ epochs.
\end{proof}

\separation*

\begin{proof}
    From Lemma \ref{lem:gridPointConfig.}, the robots reach the configuration where all the robots are located on the grid points belonging to the sector assigned to the color of the robot.
    If the robot $r$ finds both the conditions mentioned in the Section \ref{subsec:sectorToSeparartion}, the robot calculates the next destination point. The robot $r$ knows the total number of colors $k$ and the number assigned to the sector. The robot can also calculate the $rad(\lsc)$ by diameter robots. Hence, the robots can calculate the destination point correctly.
    Notice that the robots in the uppermost sector go to the outermost semicircle, and the robots in the last sector of the semicircle stay on the innermost semicircle. In this way, during this stage, all the robots can see the two diameter robots. Thus, they can calculate the radius of the innermost semicircle at any cycle.
    At any given cycle, only one robot moves from each half of the semicircle to the outer semicircle of its own half. Since all the robots have distinct destination points, they can not collide at any point.
    Even if the robot wrongly calculates its radius and goes to the wrong semicircle, the robot immediately below it signals it to come back to the innermost semicircle. The robot can see the two diameter robots, and so it can come back to the innermost semicircle and start again.
    At any given epoch, at least one robot settles on the semicircle assigned to its color. Hence, we conclude that the robots will reach the separated configuration using \textsc{SectorPartitioning-To-ConcenSemiCirc} in at most $O(n)$ epochs without any collision.

\end{proof}

From the Lemmas \ref{lemma:L1} \ref{lemma:semiCirc}, \ref{lem:gridPointConfig.}, \ref{lem:sectorPart}, and \ref{lem:finalSeparation}, we conclude the following theorem.

\begin{restatable}{thm}{final}
    \label{thm:final}
    From an arbitrary initial configuration, a collection of $n$ unconscious opaque robots reaches a separated configuration of concentric semicircles by following \textsc{Con-SemCirc-Separation} in $O(n)$ epochs without collision.
\end{restatable}

\section{Conclusion}
In this paper, we study the unconscious colored separation problem under obstructed visibility. The introduction of opacity in the robot model makes this work significant in comparison to the existing literature. We design an algorithm that achieves separation into semicircles within $O(n)$ epochs, without collisions, under the semi-synchronous scheduler.
It follows directly from our final separation structure that, if the initial configuration contains at least three robots of each color (rather than two), our algorithm also provides an immediate solution to the separation problem into concentric circles.
 To conclude, exploring whether other classical distributed problems can be solved using unconscious colored robots presents an interesting direction for future research.

\bibliographystyle{splncs04}
\bibliography{bib.bib}

\end{document}